\documentclass[twocolumn,showpacs]{revtex4}
\usepackage{amsmath,amssymb,graphicx}
\usepackage{amsfonts,amsthm}
\usepackage{amsfonts}
\usepackage{graphicx}
\usepackage{setspace}
\usepackage{bm}

\newcommand{\D}{{\mathrm{d}}}
\newtheorem{theorem}{Theorem}

\begin{document}

\title{Extremal property of a simple cycle}
\author{Alexander N. Gorban}
\affiliation{Department of Mathematics, University of Leicester, United Kingdom}
\email{ag153@le.ac.uk}

\begin{abstract}
We study systems with finite number of states $A_i$ ($i=1,\ldots, n$), which obey the
first order kinetics (master equation) without detailed balance. For any nonzero complex
eigenvalue $\lambda$ we prove the inequality $\frac{|\Im \lambda |}{|\Re \lambda |} \leq
\cot\frac{\pi}{n}$. This bound is sharp and it becomes an equality for an eigenvalue of a
simple irreversible cycle $A_1 \to A_2 \to \ldots \to A_n \to A_1$ with equal rate
constants of all transitions. Therefore, the simple cycle with the equal rate constants
has the slowest decay of the oscillations among all first order kinetic systems with the
same number of states.
\end{abstract}
 \pacs{82.20.Rp, 82.40.Bj, 02.50.Ga}

\maketitle

\paragraph*{Damped oscillations in kinetic systems.}We study systems with finite number of states $A_i$ ($i=1,\ldots, n$). A non-negative
variable $p_i$ corresponds to every state $A_i$. This $p_i$ may be probability, or
population, or concentration of $A_i$. For the vector with coordinates $(p_i)$ we use the
notation $P$.

We assume that the time evolution of $P$ obey first order kinetics (the master equation)
\begin{equation}\label{MAsterEq0}
\frac{\D p_i}{\D t}= \sum_{j, \, j\neq i} (q_{ij}p_j-q_{ji}p_i) \;\; (i=1,\ldots, n) ,
\end{equation}
where $q_{ij}$ ($i,j=1,\ldots, n$, $i\neq j$) are non-negative. In this notation,
$q_{ij}$ is the {\em rate constant} for the transition $A_j \to A_i$. Any set of
non-negative coefficients $q_{ij}$ ($i\neq j$) corresponds to a master equation.

Let us rewrite  (\ref{MAsterEq0}) in the vector notations, $\dot{P}=KP$, where the matrix
$K=(k_{ij})$ has the elements
\begin{equation}
k_{ij}=\left\{\begin{array}{cl}
q_{ij} &\mbox{ if } i\neq j;\\
-\sum_{m, \, m\neq i} q_{m i} &\mbox{ if } i = j .
\end{array}\right.
\end{equation}

The standard simplex $\Delta_n=\{P\ | \ p_i \geq 0, \, \sum_i p_i=1\}$ is {\em
forward-invariant} with respect to the master equation (\ref{MAsterEq0}) because it
preserves positivity and has the conservation law $\sum_i p_i=const$. This means that any
solution of (\ref{MAsterEq0}) $P(t)$ with the initial conditions $P(t_0)\in \Delta_n$
remains in $\Delta_n$ after $t_0$: $P(t)\in \Delta_n$ for $t\geq t_0$. One can use this
forward invariance to prove some known important properties of $K$. For example, there
exists a non-negative vector $P^*\in \Delta_n$ ($p^*_i \geq 0$) such that $KP^*=0$
(equilibrium). Indeed, any continuous map of $\Delta_n \to \Delta_n$ has a fixed point,
therefore $\exp(Kt)$ has a fixed point in $\Delta_n$ for any $t>0$. If $\exp(Kt)P^*=P^*$
for some $P^*\in \Delta_n$ and sufficiently small $t>0$, then $KP^*=0$ because
$\exp(Kt)P=P+tKP+o(t)$.

The proof that $K$ has no nonzero imaginary eigenvalues for the systems with positive
equilibria gives another simple example. We exclude one zero eigenvector and consider $K$
on the invariant hyperplane where $\sum_i p_i=0$. If $K$ has a nonzero imaginary
eigenvalue $\lambda$, then there exists a 2$D$ $K$-invariant subspace $L$, where $K$ has
two conjugated imaginary eigenvalues, $\lambda$ and $\overline{\lambda}=-\lambda$.
Restriction of $\exp(Kt)$ on $L$ is one-parametric group of rotations. For the positive
equilibrium $P^*$ the intersection $(L+P^*)\cap \Delta_n$ is a convex polygon. It is
forward invariant with respect to  the master equation (\ref{MAsterEq0}) because $L$ is
invariant, $P^*$ is equilibrium and $\Delta_n$ is forward invariant. A polygon on a plane
cannot be invariant with respect to one-parametric semigroup of rotations $\exp(Kt)$
($t\geq 0$). This contradiction proves the absence of imaginary eigenvalues. We use the
reasoning based on forward invariance of $\Delta_n$ below in the proof of the main
result.

The master equation obeys the principle of detailed balance if there exists a positive
equilibrium $P^*$ ($p^*_i>0$) such that for each pair $i,j$ ($i\neq j$)
\begin{equation}\label{detbal}
q_{ij} p^*_j=q_{ji} p^*_i .
\end{equation}
After Onsager \cite{Ons}, it is well known that for the systems with detailed balance the
eigenvalues of $K$ are real because $K$ is under conditions (\ref{detbal}) a
self-adjoined matrix with respect to the entropic inner product
$$\langle x,y\rangle=\sum_i \frac{x_i y_i}{p^*_i}$$
(see, for example, \cite{VanKampen1973,YBGE1991}).

Detailed balance is a well known consequence of microreversibility. In 1872 it was
introduced by Boltzmann for collisions \cite{Boltzmann1964}. In 1901 Wegscheider proposed
this principle for chemical kinetics \cite{Wegscheider1901}. Einstein had used this
principle for the quantum theory of light emission and absorbtion (1916, 1917). The
backgrounds of detailed balance had been analyzed by Tolman \cite{Tolman1938}. This
principle was studied further and generalized by many authors
\cite{YangHlavacek2006,GorbYabCES2012}.

Systems without detailed balance appear in applications rather often. Usually, they
represent a subsystem of a larger systems, where concentrations of some of the components
are considered as constant. For example, the simple cycle
\begin{equation}\label{simplecycle}
A_1 \to A_2 \to \ldots \to A_n \to A_1
\end{equation}
is a typical subsystem of a catalytic reaction (a catalytic cycle). The complete reaction
may have the form
\begin{equation}\label{Catcycle}
S+A_1 \to A_2 \to \ldots \to A_n \to A_1+P,
\end{equation}
where $S$ is a substrate and $P$ is a product of reaction.

The irreversible cycle (\ref{simplecycle}) cannot appear as a limit of systems with
detailed balance when some of the constants tend to zero, whereas the whole catalytic
reaction (\ref{Catcycle}) can \cite{GorbYabCES2012}. The simple cycle (\ref{simplecycle})
can be produced from the whole reaction (\ref{Catcycle}) if we assume that concentrations
of $S$ and $P$ are constant. This is possible in an open system, where we continually add
the substrate and remove the product. Another situation when such an approximation makes
sense is a significant excess of substrate in the system, $[S]\gg [A_i]$ (here we use the
square brackets for the amount of the component in the system). Such an excess implies
separation of time and the system of intermediates $\{A_i\}$ relaxes much faster than the
concentration of substrate changes.

In systems without detailed balance the damped oscillations are possible. For example,
let all the reaction rate constants in the simple cycle be equal, $q_{j+1\, j}=q_{1
n}=q>0$. Then the characteristic equation for $K$ is
$$\det(K-\lambda I)=(-q-\lambda)^n+q^{n}(-1)^{n+1}=0$$
 and $\lambda=-q+q\exp\left(\frac{2\pi i k}{n}\right) \; (k=0,\ldots ,n-1)$. For nonzero
$\lambda$,  the ratio of imaginary and real parts of $\lambda$ is
$$\frac{|\Im \lambda |}{|\Re \lambda |} =\frac{\left|\sin\frac{2\pi  k}{n} \right|}{\left|1-\cos\frac{2\pi k}{n} \right|}
=\left|\cot\frac{\pi k}{n}\right|\leq \cot\frac{\pi}{n}\, .$$ The maximal value,
$\cot\frac{\pi}{n}$, corresponds to $k=1$. For large $n$, $\cot\frac{\pi}{n}\approx
\frac{n}{\pi}$ and oscillations in the simple cycle decay rather slowly.

\paragraph*{Estimate of eigenvalues.}Let us consider the general master equation
(\ref{MAsterEq0}) without any assumption of detailed balance.

\begin{theorem}
For every nonzero eigenvalue $\lambda$ of matrix $K$
\begin{equation}
\frac{|\Im \lambda |}{|\Re \lambda |} \leq
\cot\frac{\pi}{n}
\end{equation}
\end{theorem}
\begin{proof}
Let us assume that the master equation (\ref{MAsterEq0}) has a positive equilibrium
$P^*\in \Delta_n$: for all $i=1,\ldots,n$  $p_i^*>0$ and
$$\sum_j q_{ij}p^*_j= \sum_j q_{ji} p^*_i .$$
The systems with non-negative equilibria may be considered as limits of the systems with
positive equilibria.

Let $\lambda$ be a complex eigenvalue of $K$ and let $L$ be a 2D real subspace of the
hyperplane $\sum_i p_i=0$ that corresponds to the pair of complex conjugated eigenvalues,
$(\lambda, \overline{\lambda})$. Let us select a coordinate system in the plane $L+P^*$
with the origin at $P^*$ such that restriction of $K$ on this plane has the following
matrix
$$\mathcal{K}=\left[\begin{array}{cc}
\Re \lambda &- \Im \lambda\\
\Im \lambda & \Re \lambda
\end{array}\right]\, .
$$

In this coordinate system $$\exp (t\mathcal{K})=\left[\begin{array}{cc}
\exp(t \Re \lambda) \cos (t \Im \lambda) &-\exp(t \Re \lambda) \sin(t \Im \lambda)\\
\exp(t \Re \lambda) \sin(t \Im \lambda) & \exp(t \Re \lambda) \cos (t \Im \lambda)
\end{array}\right]\, .
$$

The intersection $\mathcal{A}=(L+P^*)\cap \Delta_n$ is a polygon. It has not more than
$n$ sides because $\Delta_n$ has $n$ $(n-2)$-dimensional faces (each of them is given in
$\Delta_n$ by an equation $p_i=0$). For the transversal intersections (the generic case)
this is obvious. Non-generic situations can be obtained as limits of generic cases when
the subspace $L$ tends to a non-generic position. This limit of a sequence of polygons
cannot have more than $n$ sides if the number of sides for every polygon in the sequence
does nor exceed $n$.

Let the polygon $\mathcal{A}$ have $m$ vertices $\mathbf{v}_j$ ($m\leq n$). We move the
origin to $P^*$ and enumerate these vectors  $\mathbf{x}_i=\mathbf{v}_i-P^*$
anticlockwise (Fig~\ref{Fig:Polygon}). Each pair of vectors
$\mathbf{x}_i,\mathbf{x}_{i+1}$ (and $\mathbf{x}_m,\mathbf{x}_{1}$) form a triangle with
the angles $\alpha_i$, $\beta_i$ and $\gamma_i$, where $\beta_i$ is the angle between
$\mathbf{x}_i$ and $\mathbf{x}_{i+1}$, and $\beta_m$ is the angle between $\mathbf{x}_m$
and $\mathbf{x}_{1}$. The Sine theorem gives $\frac{|\mathbf{x}_i|}{\sin
\alpha_i}=\frac{|\mathbf{x}_{i+1}|}{\sin \gamma_i}$, $\frac{|\mathbf{x}_m|}{\sin
\alpha_m}=\frac{|\mathbf{x}_{1}|}{\sin \gamma_1}$.

Several elementary identities and inequalities hold:
\begin{equation}\label{Cond}
\begin{split}
&0 <\alpha_i,\beta_i, \gamma_i<\pi; \;\; \sum_i \beta_i=2\pi; \;\; \alpha_i+ \beta_i + \gamma_i =\pi; \\
&\prod_i \sin \alpha_i =\prod_i \sin \gamma_i \mbox{ (the closeness condition).}
\end{split}
\end{equation}
These  conditions  (\ref{Cond})  are necessary and sufficient for the existence of a
polygon $\mathcal{A}$ with these angles which is star-shaped with respect to the origin.

Let us consider the anticlockwise rotation ($\Im \lambda <0$, Fig.~\ref{Fig:Polygon}).
The case of clockwise rotations differs only in notations. For the angle $\delta$ between
$K\mathbf{x}_i$ and $\mathbf{x}_i$, $\sin \delta=-\Im \lambda$,  $\cos \delta=-\Re
\lambda$ and $\tan \delta=\frac{\Im \lambda}{\Re \lambda}$.

\begin{figure}
\centering{
\includegraphics[width=0.3\textwidth]{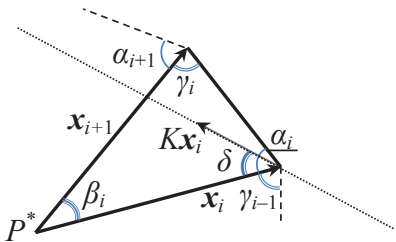}
\caption{\label{Fig:Polygon}The polygon $\mathcal{A}$ is presented as a sequence of vectors $\mathbf{x}_i$.
The angle $\beta_i$ between vectors $\mathbf{x}_i$ and $\mathbf{x}_{i+1}$ and the angles $\alpha_i$ and $\gamma_i$ of the triangle with
sides  $\mathbf{x}_i$ and $\mathbf{x}_{i+1}$ are shown.
In the Fig., rotation goes anticlockwise, i.e. $\Im \lambda <0$. In this case, the polygon $\mathcal{A}$
is invariant with respect to the semigroup $\exp (t\mathcal{K})$ ($t\geq 0$)
if and only if $\delta \leq \alpha_i$ for all $i=1, \ldots, m$, where $\delta$ is the angle between the vector field $\mathcal{K} \mathbf{x}$ and
the radius-vector $\mathbf{x}$. } }
\end{figure}

For each point $\mathbf{x} \in L+P^*$ ($\mathbf{x} \neq P^*$), the straight line
$\{\mathbf{x}+\epsilon \mathcal{K}\mathbf{x} \, | \, \epsilon \in \mathbb{R}\}$ divides
the plane $L+P^*$ in two half-plane (Fig.~\ref{Fig:Polygon}, dotted line). Direct
calculation shows that the semi-trajectory $\{\exp (t\mathcal{K})\mathbf{x}\, |\, t\geq
0\}$ belongs to the same half-plane as the origin $P^*$ does. Therefore, if $\delta \leq
\alpha_i$ for all $i=1,\ldots,m$ then the polygon $\mathcal{A}$ is forward-invariant with
respect to the semigroup $\exp (t\mathcal{K})$ ($t\geq 0$). If $\delta >\alpha_i$ for
some $i$  then for sufficiently small $t>0$  $\exp (t\mathcal{K})\mathbf{x}_i \notin
\mathcal{A}$ because $\mathcal{K}\mathbf{x}_i $ is the tangent vector to the
semi-trajectory at $t=0$.

Thus, for the anticlockwise rotation ($\Im \lambda <0$), the polygon $\mathcal{A}$ is
forward-invariant with respect to the semigroup $\exp (t\mathcal{K})$ ($t\geq 0$) if and
only if $\delta \leq \alpha_i$ for all $i=1,\ldots,m$. The maximal $\delta$
for which $\mathcal{A}$ is still forward-invariant is $\delta_{\max}=\min_i \{\alpha_i\}$. We have to find the polygon with
$m\leq n$ and the maximal value of $\min_i \{\alpha_i\}$. Let us prove that this is a
regular polygon with $n$ sides. Let us find the maximizers  $\alpha_i,\beta_i, \gamma_i$
($i=1, \ldots, m$) for the optimization problem:
\begin{equation}\label{OptimProbl}
\min_i \{\alpha_i\} \to \max\; \mbox{subject to conditions (\ref{Cond}).}
\end{equation}
For solution of this problem, all $\alpha_i$ are equal. To prove this equality, let us
mention that $\min_i \{\alpha_i\}<\frac{\pi}{2}$ under conditions (\ref{Cond}) (if all
$\alpha_i\geq \frac{\pi}{2}$ then the polygonal chain $\mathcal{A}$ cannot be closed).
Let $\min_i \alpha_i=\alpha$. Let us substitute in (\ref{Cond}) the variables $\alpha_i$
which take this minimal value by $\alpha$. The derivative of the left hand part of the
last condition in (\ref{Cond}) with respect to $\alpha$ is not zero because
$\alpha<\frac{\pi}{2}$. Assume that there are some $\alpha_j> \alpha$. Let us fix the
values of $\beta_i$ ($i=1, \ldots, m$). Then $\gamma_i$ is a function of $\alpha_i$,
$\gamma_i=\pi-\beta_i-\alpha_i$. We can use the implicit function theorem to increase
$\alpha$ by a sufficiently small number $\varepsilon>0$ and to change the non-minimal
$\alpha_j$ by a small number too, $\alpha_j\mapsto \alpha_j-\theta$;
$\theta=\theta(\varepsilon)$. Therefore, at the solution of (\ref{OptimProbl}) all
$\alpha_j = \alpha$ ($j=1,\ldots,m$).

Now, let us prove that for solution of the problem (\ref{OptimProbl}) all $\beta_i$ are
equal. We exclude $\gamma_i$ from conditions (\ref{Cond}) and write $\beta_i+\alpha<
\pi$; $0<\beta_i,\alpha$;
\begin{equation}\label{NewCOnd}
m \log\sin \alpha =\sum_i \log \sin (\beta_i+\alpha).
\end{equation}
Let us consider this equality as equation with respect to unknown $\alpha$. The function
$\log \sin x$ is strictly concave on $(0,\pi)$. Therefore, for $x_i\in (0,\pi)$
$$\log \sin\left( \frac{1}{m}\sum_{i=1}^m x_i\right)\geq \frac{1}{m}\sum_{i=1}^m\log \sin x_i$$
and the equality here is possible only if all $x_i$ are equal.  Let $\alpha^*\in
(0,\pi/2)$ be a solution of (\ref{NewCOnd}). If not all the values of $\beta_i$ are equal
and we replace $\beta_i$ in (\ref{NewCOnd}) by the average value, $\beta=\frac{2\pi}{m}$,
then the value of the right hand part of (\ref{NewCOnd}) increases and $\sin \alpha^* <
\sin (\beta+\alpha^*)$. If we take all the $\beta_i$ equal then (\ref{NewCOnd})
transforms into elementary trigonometric equation $\sin \alpha = \sin (\beta+\alpha)$.
The solution $\alpha$ of equation (\ref{NewCOnd}) increases when we replace $\beta_i$  by
the average value: $\alpha> \alpha^*$ because $\sin \alpha^* < \sin (\beta+\alpha^*)$,
$\alpha \in (0,\pi/2)$ and $\sin \alpha$ monotonically increases on this interval.  So,
for the maximizers of the conditional optimization problem (\ref{OptimProbl}) all
$\beta_i=\frac{2\pi}{m}$ and $\alpha_i=\gamma_i=\frac{\pi}{2}-\frac{\pi}{m}$. The maximum
of $\alpha$ corresponds to the maximum of $m$. Therefore, $m=n$. Finally, $\max
\{\delta\}=\frac{\pi}{2}-\frac{\pi}{n}$ and
$$\max\left\{\frac{|\Im \lambda |}{|\Re \lambda |}\right\} =
\cot\frac{\pi}{n} .$$ This is exactly the same value as for an eigenvalue of the simple
cycle of the lengths $n$ with equal rate constants, $\lambda=-q(1-\exp(\frac{2\pi
i}{n}))$.
\end{proof}

\paragraph*{Discussion.}
The simple cycle with the equal rate constants gives the slowest decay of oscillations or, in some sense, the slowest relaxation
among all first order kinetic systems with the same number of components. The extremal properties of the simple cycle with equal constants
were noticed in numerical experiments 25 years ago \cite{BochByk1987}. V.I. Bykov formulated the hypothesis that this
system has extremal spectral properties. This paper gives the answer: yes, it has.

\end{document}